\documentclass[11pt]{article}
\usepackage{epsfig}
\usepackage{subfig}
\usepackage{fontenc}

\usepackage{times}
\usepackage{amsfonts}
\usepackage{amssymb, amsmath}

\newtheorem{theorem}{Theorem}[section]
\newtheorem{definition}[theorem]{Definition}
\newtheorem{claim}[theorem]{Claim}
\newtheorem{lemma}[theorem]{Lemma}

\newtheorem{remark}[theorem]{Remark}

\newtheorem{notation}[theorem]{Notation}

\newcommand{\sq}{\hbox{\rlap{$\sqcap$}$\sqcup$}}
\newcommand{\qed}{\hspace*{\fill}\sq}
\newenvironment{proof}{\noindent {\bf Proof.}\ }{\qed\par\vskip 4mm\par}

\newcommand{\ignore}[1]{ }

\title{Equivalence of Lower Bounds on the Number of
Perfect Pairs }
\author{ V.~Ch.~Venkaiah \footnote{Corresponding author}\\
School of Computer and Information Sciences \\
University of Hyderabad \\
Professor C R Rao Road, Gachibowli \\
Hyderabad - 500 046.  \\
India.  \\ \\
K. Ramanjaneyulu \\
Department of Mathematics \\
SCR College of Engineering \\
Chilakaluripet \\ \\
Neelima Jampala \\
C. R. Rao Advanced Institute of Mathematics, Statistics, and
Computer Science \\
University of Hyderabad Campus \\
Gachibowli, Hyderabad - 500 046. \\ \\
J.~Rajendra Prasad \\
Department of Information Technology \\
P.~V.~P.~Siddartha Institute of Technology \\
Vijayawada - 520 007.\\ \\
Email: venkaiah@hotmail.com or vvcs@uohyd.ernet.in \\
kakkeraram@yahoo.co.in,neelima.jampala@gmail.com \\
rp.rajendra@rediffmail.com
\\ }
\date{\today}

\begin{document}
\maketitle

\begin{abstract}
Let $c(\mathcal{F})$ be the number of perfect pairs of
$\mathcal{F}$ and $c(G)$ be the maximum of $c(\mathcal{F})$ over
all (near-) one-factorizations $\mathcal{F}$ of $G$. Wagner showed
that for odd $n$, $c(K_{n}) \geq \frac{n*\phi(n)}{2}$ and for $m$
and $n$ which are odd and co-prime to each other, $c(K_{mn}) \geq
2*c(K_{m})*c(K_{n})$. In this note, we establish that both these
results are equivalent in the sense that they both give rise to
the same lower bound.
\end{abstract}

{\small \textbf{Keywords:} complete graph, equivalence of lower
bounds, near-one-factorization, near-one-factor of a product
graph, number of perfect pairs}

\section{Introduction}
\label{sec:intro} A one-factor of a graph $G$ of even order is a
set of edges that cover each vertex exactly once. In other words,
it is a regular spanning sub-graph of degree one \cite{DW01,
MM09}. A one-factorization of $G$ is a partition of the edge set
into a set of one-factors \cite{PKVSW08, DW01, KO05}. Analogously,
a near-one-factor of a graph $G = (V, E)$ of odd order is a
one-factor of $G \setminus {v}$ for some $v \in V$, and a
near-one-factorization of $G$ is a partition of $E$ into
near-one-factors. Our focus in this note is on near-one-factors
and near-one-factorizations.
\\

A one-factorization $\mathcal{F}$ of a complete graph $K_{2n}$ on
$2n$ vertices consists of $2n-1$ one-factors $F_{1},$ $F_{2},$ $
\cdots,$ $F_{2n-1}$. A near-one-factorization $\mathcal{F}$ of
$K_{2n-1}$ also consists of $2n-1$ near-one-factors $F_{1}, F_{2},
\cdots, F_{2n-1}$.
\\

A pair of one-factors $F_{k}$ and $F_{j}$ in a one-factorization
is said to be perfect if $F_{k} U F_{j}$ induces a Hamiltonian
cycle in $G$ \cite{BB06}. A pair of near-one-factors is called
perfect if their union is a Hamiltonian path of $G$. If every pair
of (near-) one-factors of a (near-) one-factorization is perfect
then the (near-) one-factorization is called perfect.
\\

Define $c(\mathcal{F})$ to be the number of perfect pairs of
$\mathcal{F}$ and $c(G)$ to be the maximum of $c(\mathcal{F})$
over all (near-) one-factorizations $\mathcal{F}$ of $G$
\cite{W92, WD92, W97}. Perfect one-factorization conjecture says
that for $m \geq 2$, $c(K_{2m}) = \binom{2m-1}{2}$ \cite{A77,
BMW06, DG96, K64, MR03, RR05}. This conjecture is still open
except for the case when $m$ is prime or $2m-1$ is prime or $2m
\in \{ 16,
28, 36, 40, 50, 52, 126, 170, 244, 344, 730, 1332, 1370, 1850, 2198,\\
3126, 6860, 12168, 16808, 29792 \}$ \cite{ASZ01, A73, DD06, DS89,
ISS87, KANN89, AK64, DAP95, ES91, SSDM, SSDM89, W92, WD92, W97,
Wo08, XB99}. It can be readily argued that a complete graph of
order $2m$ has a perfect one-factorization if and only if a
complete graph of order $2m-1$ has a perfect
near-one-factorization.
\\

As part of an attempt to prove the perfect one-factorization
conjecture, Wagner, in \cite{W92}, shows that for odd $n$,
$c(K_{n}) \geq n*\phi(n) / 2$, where $\phi(n)$ is the Euler's
totient function. Also proven in the same paper is that $c(K_{mn})
\geq 2*c(K_{m})*c(K_{n})$ if $m$ and $n$ are odd and are
relatively prime. Though the later result can be used with other
relevant information to arrive at a better lower bound but it is
equivalent to the former result.
\\

In this note, we show that the two results mentioned above are
equivalent in the sense that they both give rise to the same lower
bound. This equivalence is established by coming up with a
one-to-one correspondence between both the sets of
near-one-factors.

\begin{notation} $\frac{k}{r} \bmod n$ denotes $(k*\mbox{ multiplicative
inverse of }r \mbox{ with respect to }n) \bmod n$, if the
multiplicative inverse of $r$ with respect to $n$ exists.
\end{notation}

\subsection{Our Results} Main contribution of the paper is that the
two results, namely Proposition 2 and Theorem 3 of \cite{W92}, are
equivalent in the sense that they both give rise to the same lower
bound. It also extends the definition of a near-one-factor given
in Proposition 2 of \cite{W92} to one-factors and comes up with an
alternative treatment for the proposition. In addition, it renders
an algebraic description to the construction of a near-one-factor
of a product graph from those of its constituent graphs and
supplies Theorem 3 of \cite{W92} with an algebraic proof.

\subsection{Organization of the Paper} Section 2 examines the
definition of a near-one-factor given in the proof of Proposition
2 of \cite{W92}, extends it to one-factors with suitable
modifications, and provides an alternative treatment to
Proposition 2. Section 3 proposes an algebraic description to the
construction of a near-one-factor of a product graph from those of
its constituent graphs and supplies Theorem 3 of \cite{W92} with
an algebraic proof. Section 4 shows that both Proposition 2 and
Theorem 3 of \cite{W92} are equivalent in the sense that they both
give rise to the same lower bound. This is achieved by
establishing a one-to-one correspondence between the set of
near-one-factors of the product graph and the set of products of
near-one-factors of the constituent graphs. Concluding remarks are
in Section 5.

\section{One-Factors and Perfect Pairs}

This section examines the definition of a one-factor given in
\cite{W92} and explores the conditions under which two of them
form a perfect pair.
\\

Consider a graph $F_{k}$, $k \in \{0, 1, 2, \cdots, n-1\}$ on $n$
vertices with adjacency matrix $A_{k}$ that has 1 as its $i,
j^{th}$ element, where $i \neq j$ and $i + j = k \bmod n$.

\begin{claim}Let $n$ be odd. Then $F_{k}$, $k \in \{0, 1, 2, \cdots, n-1\}$ defined as above is a
near-one-factor of $K_{n}$, whose isolated vertex is $\frac{k}{2}
\bmod n$.
\end{claim}
\begin{proof}Fix $k$. Then for every $i \in \{0, 1, 2, \cdots, n-1\}$
there is a $j \in \{0, 1, 2, \cdots, n-1\}$ such that $i + j = k
\bmod n$. This is because the set $\{0, 1, 2, \cdots, n-1\}$ is
closed and each element has additive inverse in the set with
modulo $n$ addition as the operation. Also $i \neq j$ except for
$i = j = \frac{k}{2} \bmod n$; for if $i = j$, then $2i = k \bmod
n$, which implies $i = \frac{k}{2} \bmod n$. $\frac{k}{2} \bmod n$
is unique because $n$ is odd and hence multiplicative inverse of
$2$ exists. Moreover, the pair $i$ and $j$ is unique in the sense
that for a given $k$ and $i$ there is a unique $j$ with this
property. Claim follows because the first part of the discussion
implies that every vertex of $K_{n}$ occurs in $F_{k}$ and the
last statement implies that each vertex occurs exactly once.
\end{proof}
\begin{remark}When $n$ is even, only the graphs $F_{k}$, $k \in
\{1, 3, \cdots, n-1\}$ are the one-factors of $K_{n}$. For the
graph $F_{k}$, $k$ even, to be a one-factor of $K_{n}$, the
adjacency
matrix $A_{k}$ should be such that has $1$ in its $i, j^{th}$ element, where $i \neq j$ and either \\
\begin{enumerate}
\item $i, j \in \{\frac{k}{2}, \frac{n+k}{2}\}$ or \\
\item $i, j \not\in \{\frac{k}{2}, \frac{n+k}{2} \}$, and $i + j
\equiv k \bmod n$
\end{enumerate}
\end{remark}

Following lemma aids in arriving at other results.

\begin{lemma} Let $n$ be odd. Also, let $F_{k}$ and $F_{\ell}$ be
two near-one-factors of $K_{n}$ defined as above. Then the
$i^{th}$ edge of the union of these two near-one-factors, starting
from the isolated vertex of the near-one-factor $F_{k}$, is either
\begin{eqnarray*}
&&\hspace{-6in}((\frac{ik}{2} - \frac{(i-1)\ell}{2}) \bmod n,
(\frac{(i+1)\ell}{2} -
\frac{ik}{2}) \bmod n), \mbox{   if  } i \mbox{  is odd }\\ \mbox{  or } \hspace{6.75in}&&\\
&&\hspace{-6in}((\frac{i\ell}{2} - \frac{(i-1)k}{2}) \bmod n,
(\frac{(i+1)k}{2} - \frac{i\ell}{2}) \bmod n), \mbox{   if } i
\mbox{  is even }.
\end{eqnarray*}
\end{lemma}
\begin{proof} The other vertex of the edge of $F_{\ell}$ connecting the isolated
vertex of $F_{k}$, i.e.$\frac{k}{2} \bmod n$, is $(\ell -
\frac{k}{2}) \bmod n$. So, the first edge of the union starting
from the isolated vertex of $F_{k}$ is $(\frac{k}{2} \bmod n,
(\ell - \frac{k}{2})\bmod n)$. Similarly, the other vertex of the
edge of $F_{k}$ connecting the vertex $(\ell - \frac{k}{2}) \bmod
n$ is $(\frac{3k}{2} - \ell) \bmod n$. So, the second edge of the
union starting from the isolated vertex of $F_{k}$ is $((\ell -
\frac{k}{2})\bmod n, (\frac{3k}{2} - \ell) \bmod n)$. Continuing
in this way we have the third edge as $((\frac{3k}{2} - \ell)
\bmod n, (2\ell - \frac{3k}{2}) \bmod n)$, fourth edge as $((2\ell
- \frac{3k}{2})\bmod n, (\frac{5k}{2} - 2\ell) \bmod n)$, etc. In
general, the $i^{th}$ edge of the union is either
\begin{eqnarray*}
&&\hspace{-6in}((\frac{ik}{2} - \frac{(i-1)\ell}{2})\bmod n,
(\frac{(i+1)\ell}{2} - \frac{ik}{2}) \bmod n) \mbox{  for odd  } i \\ \mbox{ or } \hspace{6.75in}&&\\
&&\hspace{-6in}((\frac{i\ell}{2} - \frac{(i-1)k}{2})\bmod n,
(\frac{(i+1)k}{2} - \frac{i\ell}{2}) \bmod n) \mbox{  for even }
i.
\end{eqnarray*}
\end{proof}
\begin{lemma} Let $n$ be odd. Also, let $F_{k}$ and $F_{\ell}$ be two
near-one-factors of $K_{n}$ defined as above. Then the path
starting from the isolated vertex of either of the
near-one-factors in the union of these two does not contain a
cycle if and only if $(k - \ell)$ is relatively prime to $n$.
\end{lemma}
\begin{proof}
With out loss of generality, let us assume that the starting
vertex of the path is the isolated vertex of $F_{k}$. Then from
the previous lemma (Lemma 2.3) the $i^{th}$ edge of the path is
either
\begin{eqnarray*}
&&\hspace{-6in}((\frac{ik}{2} - \frac{(i-1)\ell}{2})\bmod n,
(\frac{(i+1)\ell}{2} - \frac{ik}{2}) \bmod n) \mbox{  if  } i \mbox{  is odd }\\ \mbox{ or } \hspace{6.75in}&&\\
&&\hspace{-6in}((\frac{i\ell}{2} - \frac{(i-1)k}{2})\bmod n,
(\frac{(i+1)k}{2} - \frac{i\ell}{2}) \bmod n) \mbox{  if } i
\mbox{  is even }.
\end{eqnarray*}
For there to be a cycle on this path, there must exist two
distinct positive integers $i$ and $j$ such that either
\begin{eqnarray*}
&&\hspace{-6in}(\frac{ik}{2} - \frac{(i-1)\ell}{2}) \bmod n =
(\frac{jk}{2} -
\frac{(j-1)\ell}{2}) \bmod n \\
\mbox{ or } \hspace{6.75in} &&\\
&&\hspace{-6in}(\frac{i\ell}{2} - \frac{(i-1)k}{2}) \bmod n =
(\frac{j\ell}{2} - \frac{(j-1)k}{2}) \bmod n
\end{eqnarray*}
But they are equivalent to
\begin{eqnarray*}
&& (\frac{(i-j)}{2})(k - \ell) \equiv 0 \bmod n \\
\Leftrightarrow && (i - j)(k - \ell) \equiv 0 \bmod n
\hspace{0.2cm}(\because \mbox{multiplicative inverse of }
2\mbox{ with respect to }n\mbox{ exists})\\
\Leftrightarrow && i = j \hspace{0.2cm}(\because (k - \ell) \mbox{
is relatively prime to } n \mbox{ and both } i \mbox{ and } j
\mbox{ are less than } n)
\end{eqnarray*}
\end{proof}
\begin{lemma}Let $n$ be odd. Then two near-one-factors $F_{k}$ and
$F_{\ell}$ defined as above is a perfect pair if and only if $k -
\ell$ is relatively prime to $n$.
\end{lemma}
\begin{proof} Consider the path starting from the isolated vertex of $F_{k}$
in the union of the two near-one-factors. As there is no cycle on
this path, it follows that the other end of the path must be the
isolated vertex $\frac{\ell}{2} \bmod n$ of $F_{\ell}$. So,
\begin{eqnarray*}
&&\frac{\ell}{2} \bmod n = (\frac{(i+1)k}{2} - \frac{i\ell}{2})
\bmod n \\
\Leftrightarrow && \frac{\ell}{2} + \frac{i\ell}{2} -
\frac{(i+1)k}{2}
\equiv 0 \bmod n \\
\Leftrightarrow && \frac{(i+1)}{2}(\ell - k) \equiv 0 \bmod n \\
\Leftrightarrow && (i+1)(\ell - k) \equiv 0 \bmod n
\hspace{0.2cm}(\because 2 \mbox{
has multiplicative inverse}) \\
\Leftrightarrow && (i+1) \mbox{  is a multiple of } n
\hspace{0.2cm}(\because k - \ell
\mbox{  is relatively prime to  } n) \\
\Leftrightarrow && i + 1 = n \hspace{0.2cm}(\because \mbox{ length
of any path in the
union of two near-one-factors is less than } n) \\
\Leftrightarrow && i = n - 1.
\end{eqnarray*}
That is the path connecting the isolated vertices is a hamiltonian
path. Hence the claim.
\end{proof}
\begin{remark}From the discussion above it follows that a near-one-factor
$F_{k}$ forms a perfect pair with another near-one-factor
$F_{\ell}$ if and only if $k - \ell$ is relatively prime to $n$.
For a fixed $k$, the number of such $\ell$'s is equal to
$\phi(n)$.(This may be proved by observing that a pair of integers
$\ell$ and $c*n - \ell$, $c$ is some integer, is either both
relatively prime to $n$ or both not.) Hence the number of perfect
pairs with one of the near-one-factor in the pair being $F_{k}$ is
$\phi(n)$. So, the total number of perfect pairs is
$\frac{(n*\phi(n))}{2}$. Note that this is the result of the
Proposition 2 of \cite{W92}).
\end{remark}

\section{One-Factors and Product Graphs}

This section provides an algebraic description of the construction
of a near-one-factor of a product graph from those of its
constituent graphs. It also analyzes the conditions under which
two near-one-factors of a product graph form a perfect pair and
supplements Theorem 3 of \cite{W92} with an algebraic proof.

\begin{definition}\textsl{\emph{[Product Graph]}} The product graph $M \times N$ of its constituent graphs $M$
and $N$ is defined as follows:
\begin{enumerate}
\item Vertex set, $V(M \times N)$, of the product graph $M \times
N$ is the cartesian product of the vertex sets $V(M)$ and $V(N)$
of its constituent graphs $M$ and $N$ respectively. That is $V(M
\times N)$ $=$ $V(M) \times V(N)$. \item The edge set of the
product graph, $E(M \times N)$, is $\{(v, w), (v', w')\} \in E(M
\times N)$ if and only if $\{v, v'\} \in E(M)$ and $\{w, w'\} \in
E(N)$.
\end{enumerate}
\end{definition}
\begin{claim} Let $s$ and $t$ be odd positive integers.
Also let $0 \leq i,i',k < s$ and $0 \leq j,j',\ell < t$.
Further, let $G_{k}$ and $H_{\ell}$ denote near-one-factors of the
complete graphs $K_{s}$ and $K_{t}$ respectively. Define $D_{k,
\ell}$ to be a square 0, 1 matrix whose rows and columns are
indexed by ordered pairs $(i, j)$ such that the element in $(i,
j)^{th}$ row and $(i', j')^{th}$ column is 1 if and only if $(i,
j) \neq (i', j')$, $(i + i') \bmod s = k$, and $(j + j') \bmod t =
\ell$. Then $D_{k, \ell}$ is an adjacency matrix of the product
graph $G_{k} \times H_{\ell}$.
\end{claim}
\begin{proof} We have the vertex set, $V(G_{k} \times H_{\ell})$, of the product
graph $G_{k} \times H_{\ell}$ as $\{(i, j) : i \in V(G_{k}) \mbox{
and } j \in V(H_{\ell}) \}$. Also, $\{(i, j), (i', j')\}$ is an
edge in the product graph if and only if both $\{i, i'\}$ and
$\{j, j'\}$ are edges in the graphs $G_{k}$ and $H_{\ell}$
respectively. But $\{i, i'\}$ and $\{j, j'\}$ are edges in their
respective graphs if and only if $i \neq i'$, $j \neq j'$, $(i +
i') \bmod s = k$, and $(j + j') \bmod t = \ell$. Hence the claim.
\end{proof}

\begin{claim}Let $n = s \times t$. Then $D_{k, \ell}$, where $0 \leq k < s$ and $0 \leq \ell < t$,
defined in Claim 3.2 is an adjacency matrix of a near-one-factor
of $K_{n}$, whose isolated vertex is $((\frac{k}{2}) \bmod s,
(\frac{\ell}{2}) \bmod t)$. That is, the product graph $G_{k}
\times H_{\ell}$ is a near-one-factor of the complete graph
$K_{n}$.
\end{claim}
\begin{proof} Fix $k$ and $\ell$ and argue as in Claim 2.1 by treating
$k$ as the ordered pair $(k, \ell)$, $i$ as $(i, j)$, $j$ as $(i',
j')$, and $n$ as $(s, t)$.
\end{proof}

\begin{remark} Let $G_{k}$, $G_{k'}$ denote a pair of
near-one-factors of the complete graph $K_{s}$ and $H_{\ell}$,
$H_{\ell'}$ denote a pair of near-one-factors of the complete
graph $K_{t}$. From the previous claim we have $G_{p} \times
H_{q}$, $p \in \{k, k'\}$ and $q \in \{\ell, \ell'\}$, are
near-one-factors of the complete graph $K_{st}$.
\end{remark}
\begin{claim} Let $G_{k}$, $G_{k'}$, $H_{\ell}$, and $H_{\ell'}$
 be as defined in the above remark. Then a pair of
 near-one-factors from the set $\{G_{p} \times H_{q} : p \in \{k, k'\}
 \mbox{ and } q \in \{\ell, \ell'\}\}$ of the complete graph $K_{st}$
 is perfect if and only if both the pairs, namely $G_{k}$, $G_{k'}$ and
 $H_{\ell}$, $H_{\ell'}$, are perfect for the corresponding
 complete graphs.
\end{claim}
\begin{proof} Consider two near-one-factors $G_{k} \times H_{\ell}$
and $G_{k'} \times H_{\ell'}$ of the complete graph $K_{st}$.
Arguing as in Section 2, these two near-one-factors can be shown
to form a perfect pair if and only if $(k - k')$ is relatively
prime to $s$ and $(\ell - \ell')$ is relatively prime to $t$. But
from Lemma 2.5, it means that both the pairs namely $G_{k}$,
$G_{k'}$ and $H_{\ell}$, $H_{\ell'}$ are perfect. Hence the claim.
\end{proof}
\begin{remark}Since $\gcd((x - x), s) = s$, it follows from the proof of the above
claim that there are two perfect pairs, namely $G_{k} \times
H_{\ell}$, $G_{k'} \times H_{\ell'}$ and $G_{k'} \times H_{\ell}$,
$G_{k} \times H_{\ell'}$ of $K_{st}$ for every
 perfect pair $G_{k}$, $G_{k'}$ of $K_{s}$ and perfect pair
 $H_{\ell}$, $H_{\ell'}$ of $K_{t}$. So, the number of perfect
 pairs of $K_{st}$ is more than or equal to twice
 the number of perfect pairs of $K_{s}$ times the number of
 perfect pairs of $K_{t}$. That is, $c(K_{st}) \geq 2 * c(K_{s}) *
 c(K_{t})$. Note that this is the result of Theorem 3 of \cite{W92}.
\end{remark}

\section{Equivalence of Lower Bounds}

In the previous section, we have shown that the product of
near-one-factors of constituent graphs is a near-one-factor of the
product graph. We now show that, under certain mild conditions, a
near-one-factor of a product graph is the product of
near-one-factors of its constituent graphs. This is achieved by
establishing a one-to-one correspondence between the set of
near-one-factors of the product graph and the set of products of
near-one-factors of its constituent graphs. So, it implies that
the lower bounds obtained in Sections 2 and 3 are equivalent.

\begin{lemma} Let $s$ and $t$ be odd positive integers and are
co-prime to each other. Also, let $n = s \times t$. Further, let
$A_{p}$, $0 \leq p < n$, and $D_{k, \ell}$, $0 \leq k < s$, $0
\leq \ell < t$, be 0, 1 matrices defined in Section 2 and Claim
3.2 respectively. Then there is a one-to-one correspondence
between the sets $\{A_{p} : 0 \leq p < n\}$ and $\{D_{k, \ell} : 0
\leq k < s, 0 \leq \ell < t\}$.
\end{lemma}
\begin{proof}Define a mapping from the set of matrices $\{A_{p} : 0 \leq p <
n\}$ to the set of matrices $\{D_{k, \ell} : 0 \leq k < s, 0 \leq
\ell < t\}$ such that the matrix $A_{p}$ gets mapped to $D_{k,
\ell}$ if and only if $p \bmod s = k$ and $p \bmod t = \ell$.
Since $s$ and $t$ are co-prime to each other, by Chinese Remainder
Theorem, this mapping is a one-to-one correspondence.
\end{proof}
\begin{remark}Since $\{A_{p} : 0 \leq p < n\}$ denote the adjacency
matrices of a set of near-one-factors of $K_{n}$ and $\{D_{k,
\ell} : 0 \leq k < s, 0 \leq \ell < t\}$ also denote the adjacency
matrices of another set of near-one-factors of $K_{n}$, it follows
from the above lemma that these two sets of near-one-factors are
one and the same. So, the number of perfect pairs in both the sets
of near-one-factors is same.
\end{remark}
\section{Conclusions}This note establishes that the apparently two different lower bounds
derived in \cite{W92} are one and the same. It also extends the
definition of a near-one-factor given in Proposition 2 of
\cite{W92} to one-factors and comes up with an alternative
treatment for the proposition. In addition, it renders an
algebraic description to the construction of a near-one-factor of
a product graph from those of its constituent graphs and supplies
Theorem 3 of \cite{W92} with an algebraic proof.

\bibliographystyle{plain}

\end{document}